\documentclass[reprint,aps,pra,]{revtex4-2}

\usepackage{graphicx}
\graphicspath{ {./img/} }
\usepackage{dcolumn}
\usepackage{tikz}
\usepackage{taisanul}
\usepackage{braket}
\usepackage{bm}
\usepackage{hyperref}

\begin{document}

\preprint{APS/123-QED}

\title{Aspects of Quantum Energy Teleportation}

\author{Taisanul Haque}
    \email{taisanul.haque@stud.uni-goettingen.de}

\affiliation{%
 Institute of Theoretical Physics, University of Goettingen
}%

\date{\today}

\begin{abstract}
In this work, we explore quantum energy teleportation (QET) protocols, focusing on their behavior at finite temperatures , in ground and excited states. We analyze the role of entanglement as a resource for QET, particularly in thermal states, and compare the performance of QET across these initial states. We then introduce a method to extract ground-state energy through a protocol that employs only quantum measurements, local operations, and classical communication (LOCC), without requiring the ground state to be quantum correlated either through entanglement or quantum discord. To illustrate this, we propose a minimal model comprising two interacting qubits. These findings indicate that, in addition to the established QET framework where quantum correlation serves as a resource, it is possible to extract energy from an product ground state. This broadens the scope of QET's applicability across diverse quantum systems.
\end{abstract}

\maketitle


\section{\label{sec:level1}Introduction}
In the seminal paper \cite{PhysRevLett.70.1895} by Bennett et al. present the first theoretical proposal for quantum information teleportation (QIT), where an unknown quantum state is transferred from one location to another using local operation and classical communication(LOCC), and previously shared entanglement (Einstein-Podolsky-Rosen channels). Nevertheless, it was demonstrated in \cite{hotta2011quantumenergyteleportationintroductory} that the QIT alone is inadequate for transferring energy. In fact, it has been shown\cite{wang2024quantumenergyteleportationversus} that there is a trade off relationship between QET and QIT.
The Quantum Energy Teleportation (in short QET) protocol originally proposed by  Masahiro Hotta \cite{Hotta_2008,Hotta_2009,Hotta_2010} is fundamentally different from classical energy transport methods. It utilizes zero-point \textcolor{black}{or ground state energy}, which is the lowest possible energy that a quantum mechanical system may have. Unlike classical systems, where zero energy signifies an absence of motion, quantum systems maintain non-zero energy due to zero-point fluctuations \textcolor{black}{because of non-commutativity of quantum variables}. These fluctuations are crucial \textcolor{black}{\cite{PhysRevA.89.012311, PhysRevD.78.045006}} for QET protocol as they represent the potential energy that can be teleported. \textcolor{black}{The key point in QET is that the ground state energy is protected from direct exploitation in the sense that one cannot (say, Alice) extract energy using direct measurement or local unitary operations. The QET protocol requires a special class of measurement operators that only locally inject energy at Alice's site and do not disturb the other region (say, Bob's site), where energy is to be extracted. After the measurements, the whole system is in an excited state, but the energy density should remain unaffected at Bob's site. Thus, in QET, there should exist LOCC, which enables lowering the local energy density at Bob's site, i.e., the local energy density is negative, although the overall energy density remains non-negative. This is allowed in quantum systems \cite{Hotta_2009}, and we refer to this as local zero-point energy extraction.} Therefore, through a QET protocol, the quantum fluctuations in a vacuum can be harnessed to extract and transfer energy via entanglement correlations.  This idea has been generalized to other quantum systems such as QET in harmonic chain model\cite{PhysRevA.82.042329}, spin chain model\cite{Hotta_2009}, trapped ion model\cite{PhysRevA.80.042323}, quantum hall system\cite{PhysRevA.84.032336}, in black hole\cite{PhysRevD.81.044025}. Furthermore, QET was recently realized in the experiment \cite{Rodriguez-Briones:2022jla}, and it was demonstrated in quantum hardware \cite{PhysRevApplied.20.024051}.

\textcolor{black}{However, it was emphasized in} \cite{Frey_2013} for their a specific minimal model that entanglement is not a necessary ingredient for QET by studying it on a thermal state; It is argued that quantum discord is sufficient for QET in that model. Moreover, it was argued \textcolor{black}{in the discussion and outlook section} that a product state cannot support QET \textcolor{black}{in  the same model}. \textcolor{black}{In this paper, we have explored various aspects of QET in different parametric regimes $(T,B)$ for a fixed $\alpha$ in 
a different model. We confirmed the claim about the non-necessity of entanglement by demonstrating QET in the presence of discord and the absence of entanglement in certain parametric regimes of $(T,B)$, as shown in Figure \ref{fig:1}. Furthermore, we have shown that there exist regimes in $(T,B)$ where no quantum discord is present, yet QET remains successful. This has been analytically demonstrated by proving in Section \ref{sec:level5} that neither quantum nor classical correlations are necessary, as a product ground state can still support QET. However, the original form of QET begins with some correlation in the initial state. Thus, we cannot classify this as QET and have instead termed it QEE. Additionally, in Section \ref{sec:level4}, we discuss QET in a non-passive state, which does not violate the conditions of QET except for the initial state being non-passive.} 

\textcolor{black}{We observe that when the ground state is entangled, which occurs for $B < \alpha$ in our thermal state, increasing the temperature thermal entanglement decreases, and energy extraction through QET also decreases. On the other hand, when the ground state is a product state, which occurs for $B > \alpha$, entanglement increases with temperature as the ground state mixes with other entangled eigenstates, yet energy extraction decreases. At the critical point $\alpha = B$, energy extraction vanishes despite the presence of nonzero entanglement.}

\subsection{\label{sec:subA}The generalized protocol}
In a usual QET protocol, a local measurement is performed at a site A (Alice's site), which collapses the quantum state and reveals information about the system's energy distribution. In this measurement process, some energy is injected into the site A.  The classical information about the measurement outcome is subsequently transmitted to a different location i.e., point B (Bob's site). Based on the received information, a local operation is performed at point B to extract energy from the ground state.

Here, we propose a generalized protocol in the sense that it relaxes the condition requiring the initial state in QET to be quantum correlated through entanglement or discord. Suppose we have a generic quantum system with a Hamiltonian $\bm{H}$. We consider a general state $\bm{\rho}^{AB}$ prepared from $\bm{H}$, which can be entangled ground or excited state, thermal state and so on. The protocol is done in the following way: Alice performs a measurement on her system with measurement operators $\{M_i\}$ such that $\sum_iM_i^\dagger M_i=\mathbf{I}$.
The probability of the outcome $i$ is given by the Born rule, $p_i=\tr(\mathbf{E}_i\bm{\rho})$ in which $\mathbf{E}_i=M_i^\dagger M_i$. After the measurement, the state $\bm{\rho}^{AB}$ is updated to $\frac{M_i\bm{\rho}^{AB}M_i^\dagger}{p_i}$. Now if Alice lost tracks of the measurement, then resulting ensemble description is given by
$$\bigg\{p_i, \frac{M_i\bm{\rho}^{AB}M_i^\dagger}{p_i}\bigg\}_i$$
Therefore on an average, the density operator after the measurement is given by
\begin{align}
    \bm{\rho}_A=\sum_i p_i \frac{M_i\bm{\rho}^{AB}M_i^\dagger}{p_i}=\sum_iM_i\bm{\rho}^{AB}M_i^\dagger\label{eq:1}
\end{align}
In the measurement process,  Alice infuses / extracts some energy to/from the system depending on the state being used to start the protocol, given by
\begin{align}
    \Delta E_{\text{Alice}}=\tr(\bm{\rho}_A \bm{H})-\tr(\bm{\rho}^{AB} \bm{H})\label{eq:2}
\end{align}
Once the measurement is done, Alice convey the outcome of her measurement to Bob with a single round of classical communication , and Bob performs a local unitary operation(LOCC) $\bm{U}_i$ on his system , thus on an average, the final state after LOCC is given by
\begin{align}
    \bm{\rho}_B=\sum_i \bm{U}_i M_i\bm{\rho}^{AB}M_i^\dagger \bm{U}_i^\dagger\label{eq:3}
\end{align}

After this protocol, Bob can extracts some energy given by
\begin{align}
    \Delta E_{\text{extract}}=\tr(\bm{\rho}_A \bm{H})-\tr(\bm{\rho}_B \bm{H})\label{eq:4}
\end{align}
Whenever the equation\eqref{eq:4} is positive.

\section{\label{sec:level2} The minimal QET model}

\subsection{\label{sec:subIIA}State Preparation}
We have considered one of the simplest bipartite system that contains at least one entangled eigenstates i.e., a model with two interacting qubits  with a interaction parameter $\alpha$ in an external magnetic field $B$ along the $z-$direction. This model is also known as Heisenberg XY model. The model is described as 
$$\bm{H}=\frac{B}{2}(\sigma^{A}_z+\sigma^{B}_z)+\alpha(\sigma^{A}_{+}\ot\sigma^{B}_{-}+\sigma^{A}_{-}\ot\sigma^{B}_{+})$$
Where $\sigma^{i}_{\pm}=\frac{\sigma^{i}_x\pm i\sigma^{i}_y}{2}$ for $i=A,B$. Now we attach the whole system with a thermal bath at a temperature $T$ 


So the initial mixed state for the entire system is described by the density matrix 
\begin{align}
    \bm{\rho}^{AB}=\frac{e^{-\beta \bm{H}}}{\tr(e^{-\beta \bm{H}})}\label{eq:5}
\end{align}

where $\beta=\frac{1}{k_B T}$ is inverse temperature parameter and $k_B:=1$ is the Boltzmann constant.
The partition function is given by $Z=\tr(e^{-\beta \bm{H}})=\sum_{n}\braket{e_n|e^{-\beta\bm{H}}|e_n}$ where $e_n$ are the eigenstates of $\bm{H}$.
Also we can see that the system lives in the Hilbert space spanned by
$\{0,1\}^{\ot 2}$, thus the dimension of the Hilbert space is $\dim\spn{0,1}^{\ot 2}=4$. In the below table, we have described all the possible distinct energy states along with the number of degeneracy.  
\begin{center}
\begin{tabular}{ c c c c   }\label{tab:tab1}
 \text{Level}& \text{States}& \text{Degeneracy} &\text{Energy}\\ 
\hline
 $E_{00}:$& $\ket{00}$& $1$ & $-B$ \\
 $E_{11}:$& $\ket{11}$& $1$ & $+B$ \\
 $E_{+}:$& $\ket{\psi^+}:=\frac{\ket{01}+\ket{10}}{\sqrt{2}}$& $1$ & $\alpha$ \\
 $E_{-}:$& $\ket{\psi^-}:=\frac{\ket{01}-\ket{10}}{\sqrt{2}}$& $1$ & $-\alpha$ \\
\end{tabular}
\end{center}
From the above table we can observe that the states $\ket{\psi^\pm}$ are entangled states.
Now to keep the Hamiltonian positive semi-definite, we can add a suitable term $\epsilon$ to the original Hamiltonian $\bm{H}$ which does not affect the dynamics of the system.

Since $\{\ket{e_n}\}_n=\{\ket{00},\ket{11},\ket{\psi^\pm}\}$ are complete orthonormal eigenbasis of the Hamiltonian $\bm{H}$ i.e.,  $\bm{H}\ket{e_n}=E_n\ket{e_n}$ and $\sum_n\ket{e_n}\bra{e_n}=\mathbf{I}$. Thus we can rewrite the initial density matrix from the equation\eqref{eq:5} as follows:
\begin{align}
    \bm{\rho}^{AB}=\sum_np_n\ket{e_n}\bra{e_n}\label{eq:6}
\end{align}
Where $p_n:=\frac{e^{-\beta E_n}}{Z}$ denotes the thermal probability of the eigenstate $\ket{e_n}$   and $Z$ is the partition function given by
\begin{align}
    Z=\sum_ng(E_n)e^{-\beta E_n}=2(\cosh(\beta \alpha)+\cosh(\beta B))\label{eq:7}
\end{align}
Where $g(E_n)$ counts the number of degenerate states with energy $E_n$. We assume, $B\le \alpha$ to keep $\ket{\psi^{-}}$ as the entangled ground state. We can observed that at a low temperature limit i.e $T\ra 0$, the inverse temperature $\beta\ra \inn$, consequently $p_{-}$ survives as $E_{-}+\epsilon=0$ ,and $E_n+\epsilon>0$ for all $n$. Therefore, only the ground state survives, thus in this limit QET protocol for the ground state can be recovered.

\subsection{\label{sec:subB}Measurement on site A}
In this subsection we will begin with a projective measurement scheme for Alice. Suppose Alice makes projective measurements on her spin qubit system,  which is given by $M_A(k)=\frac{1}{2}(\mathbf{I}+k\sigma_x^A)$ for $k\in\{\pm 1\}$. Clearly we see that $M_A(k)=M^\dagger_A(k)$ and $M_A(k)M^\dagger_A(k)=M_A(k)$. Moreover it satisfy $\sum_{k\in \{\pm 1\}}M_A(k)=\mathbf{I}$. The measurement outcome through Alice's measurements are $\pm 1$. After the measurements, the state $\bm{\rho}^{AB}$ takes the similar form as in the equation \eqref{eq:1} i.e.,
\begin{align}
    \bm{\rho}_A=\sum_{k\in \{\pm 1\}}\sum_np_n\ket{\Psi_{A,n}(k)}\bra{\Psi_{A,n}(k)}\label{eq:8}
\end{align}

Where $\ket{\Psi_{A,n}(k)}:=M_A(k)\ket{e_n}$ for all $k\in \{\pm 1\}$ and $n\in \{00,11,+,-\}$, are given below:
\begin{align}
    \ket{\Psi_{A,00}(k)}&=\frac{1}{2}\bigg[\ket{e_{00}}+\frac{k}{\sqrt{2}}(\ket{e_+}-\ket{e_-})\bigg]\label{eq:9}\\
    \ket{\Psi_{A,11}(k)}&=\frac{1}{2}\bigg[\ket{e_{11}}+\frac{k}{\sqrt{2}}(\ket{e_+}+\ket{e_-})\bigg]\label{eq:10}\\
    \ket{\Psi_{A,+}(k)}&=\frac{1}{2}\bigg[\ket{e_{+}}+\frac{k}{\sqrt{2}}(\ket{e_{11}}+\ket{e_{00}})\bigg]\label{eq:11}\\
    \ket{\Psi_{A,-}(k)}&=\frac{1}{2}\bigg[\ket{e_{-}}+\frac{k}{\sqrt{2}}(\ket{e_{11}}-\ket{e_{00}})\bigg]\label{eq:12}
\end{align}
Now using equations \eqref{eq:9}-\eqref{eq:12}, we can compute the energy at site A through the measurement process as:
\begin{align}
    E_A=\tr(\bm{H}\bm{\rho}_A)=\sum_k\sum_np_n\braket{\Psi_{A,n}(k)|\bm{H}|\Psi_{A,n}(k)}\label{eq:13}
\end{align}
From the equation \eqref{eq:13}, we denote $E_{A,n}=\braket{\Psi_{A,n}(k)|\bm{H}|\Psi_{A,n}(k)}$ for all $n\in\{00,11,+,-\}$. Moreover we found that $E_{A,00}=-E_{A,11}=-\frac{B}{4}$ and $E_{A,+}=-E_{A,-}=\frac{\alpha}{4}$
Therefore equation \eqref{eq:13} becomes:
\begin{align}
    E_A=-B\frac{\sinh(\beta B)}{Z}-\alpha\frac{\sinh(\beta \alpha)}{Z}\label{eq:14}
\end{align}
Thus from the equation \eqref{eq:2}, the energy injected in the measurement is given by
\begin{align}
    \boxed{\Delta E_{\text{inf}}=B\frac{\sinh(\beta B)}{Z}+\alpha\frac{\sinh(\beta \alpha)}{Z}}\label{eq:15}
\end{align}
\subsection{\label{sec:subC}
LOCC and energy extraction for B
}
Measurements by Alice have been performed in the above subsection. She now communicates her results through a classical communication channel. The assumption for the time scale here is much shorter than the time scale for energy diffusion after the measurement. Based on the information about the measurement outcomes, Bob performs a local unitary operation on his site. Suppose that for each measurement $k$, Bob chooses a unitary operation $U_B(\theta,k) := \cos(\theta)\mathbf{I} + ik\sigma^B_y\sin(\theta)$. In Appendix \ref{sec:B}, we have shown that this choice for $U_B(\theta,k)$ is indeed the optimal choice, maximizing energy teleportation. In this subsection, we determine an optimized value of $\theta = \theta_0$ such that the corresponding optimized local unitary operation $U_B(\theta_0,k)$ maximizes energy teleportation from Alice to Bob.

From the equation \eqref{eq:3}, the density matrix takes the following form after the local operation by Bob, 
\begin{align}
    \bm{\rho}_{B}=\sum_{k=\pm 1}\sum_np_n\ket{\Psi_{B,n}(k)}\bra{\Psi_{B,n}(k)}\label{eq:16}
\end{align}
where the transformed eigenstates after Bob's action are $\ket{\Psi_{B,n}(k)}:=U_B(\theta,k)\ket{\Psi_{A,n}(k)}$.

Using equations \eqref{eq:9}-\eqref{eq:12}, we explicitly calculate all the states $\ket{\Psi_{B,n}(k)}$ below:
\begin{align}
    \ket{\Psi_{B,00}(k)}&=\frac{\cos(\theta)}{2}\ket{e_{00}}+\frac{\sin(\theta)}{2}\ket{e_{11}}\notag\\
    &\hspace{1cm}+k\frac{(\cos(\theta)+\sin(\theta))}{2\sqrt{2}}\ket{e_+}\notag\\
    &\hspace{2cm}-k\frac{(\cos(\theta)-\sin(\theta))}{2\sqrt{2}}\ket{e_-}\label{eq:17}\\
    \ket{\Psi_{B,11}(k)}&=\frac{\cos(\theta)}{2}\ket{e_{11}}-\frac{\sin(\theta)}{2}\ket{e_{00}}\notag\\
    &\hspace{1cm}+k\frac{(\cos(\theta)-\sin(\theta))}{2\sqrt{2}}\ket{e_+}\notag\\
    &\hspace{2cm}+k\frac{(\cos(\theta)+\sin(\theta))}{2\sqrt{2}}\ket{e_-}\label{eq:18}\\    
    \ket{\Psi_{B,+}(k)}&=\frac{\cos(\theta)}{2}\ket{e_{+}}+\frac{\sin(\theta)}{2}\ket{e_{-}}\notag\\
    &\hspace{1cm}+k\frac{(\cos(\theta)+\sin(\theta))}{2\sqrt{2}}\ket{e_{11}}\notag\\
    &\hspace{2cm}+k\frac{(\cos(\theta)-\sin(\theta))}{2\sqrt{2}}\ket{e_{00}}\label{eq:19}\\
    \ket{\Psi_{B,-}(k)}&=\frac{\cos(\theta)}{2}\ket{e_{-}}-\frac{\sin(\theta)}{2}\ket{e_{+}}\notag\\
    &\hspace{1cm}+k\frac{(\cos(\theta)-\sin(\theta))}{2\sqrt{2}}\ket{e_{11}}\notag\\
    &\hspace{2cm}+k\frac{(\sin(\theta))+\cos(\theta)}{2\sqrt{2}}\ket{e_{00}}\label{eq:20}
\end{align}
Thus we have all the information to calculate the energy in the Bob's system after LOCC. From equation \eqref{eq:16} and using equations \eqref{eq:17}-\eqref{eq:20}, the expression for $E_B$ is given by,
\begin{align}
    E_B=\tr(\bm{H}\bm{\rho}_B)=\sum_k\sum_np_n\braket{\Psi_{B,n}(k)|\bm{H}|\Psi_{B,n}(k)}\label{eq:21}
\end{align}
Moreover, from the above equation, we denote $E_{B,n}=\braket{\Psi_{B,n}(k)|\bm{H}|\Psi_{B,n}(k)}$ for all $n$. Thus we found $E_{B,00}=-E_{B,11}=-\frac{1}{4}(\cos(2\theta)B-\sin(2\theta)\alpha)$ and $E_{B,+}=-E_{B,-}=\frac{1}{4}(\cos(2\theta)\alpha+\sin(2\theta)B)$
Therefore using the above expressions for $E_{B,n}$, the equation \eqref{eq:21} becomes,
\begin{align}
    E_B=\sum_k\sum_np_nE_{B,n}&=-2l(p_{00}-p_{11}) -2m (p_{-}-p_{+})\notag\\
    &=-4l\frac{\sinh(\beta B)}{Z}-4m\frac{\sinh(\beta \alpha)}{Z}\label{eq:22}
\end{align}
Where we have denoted $l=\frac{1}{4}(\cos(2\theta)B-\sin(2\theta)\alpha)$, and $m=\frac{1}{4}(\cos(2\theta)\alpha+\sin(2\theta)B)$

\subsubsection{Energy optimization for Bob's action\label{sec:subsub1}}
After the LOCC, Bob observes that his system loses some energy which is given by
$\Delta E_{\text{extract}}=|E_B-E_A|$ i.e., from the equation \eqref{eq:14} and equation \eqref{eq:22}
\begin{align}
    E_B-E_A&=(-4l+B)\frac{\sinh(\beta B)}{Z}+(-4m+\alpha)\frac{\sinh(\beta \alpha)}{Z}\label{eq:23}
\end{align}

Now we optimize the unitary local operation that Bob has performed such that energy extraction by Bob is maximum. Let us denote the equation \eqref{eq:23} by $F(t=2\theta):=E_B-E_A$, thus after simplification we get
\begin{align}
    F(t)&=\left(B(1 - \cos(t)) + \alpha \sin(t)\right) \frac{\sinh(\beta B)}{Z}\notag\\
    &\hspace{1cm}+ \left(\alpha(1 - \cos(t)) - B\sin(t)\right) \frac{\sinh(\beta \alpha)}{Z}\notag
\end{align}
After setting the derivative of $F(t)$ with respect to $t$ to zero i.e., $\frac{dF}{dt}=0$, we find the critical point of $F(t)$ at $t=t_0\iff \theta=\theta_0$ by solving:
\begin{align}
    \boxed{\tan(t_0) = \frac{ B\sinh(\beta \alpha)-\alpha\sinh(\beta B)}{B\sinh(\beta B) + \alpha\sinh(\beta \alpha)}}\label{eq:5.57}
\end{align}
In the appendix \ref{sec:A}, we have explicitly shown that at $\theta=\theta_0$, indeed $F$ takes the minimum value and the minimum value is negative. Thus energy extraction by Bob $|F(t_0)|$ is maximum.
Let us denote $\tan(t_0):=\frac{q}{p}$, where $p:=B\frac{\sinh(\beta B)}{Z}+\alpha\frac{\sinh(\beta \alpha)}{Z}$ and $q:=B\frac{\sinh(\beta \alpha)}{Z}-\alpha\frac{\sinh(\beta B)}{Z}$. Thus, the simplest expression for the minimum is:
\begin{align}
    \boxed{E_B - E_A = p(1 - \cos(t_0)) - q\sin(t_0)=p-\sqrt{p^2+q^2}<0}\label{eq:25}
\end{align}
\textcolor{black}{The above equation \eqref{eq:25} suggests that Bob can extract some energy from his site. To further illustrate this, we have plotted a colormap in the Figure \ref{fig:1} to show how much energy Bob can extract. We observe that for every $\alpha$, there exist regions in $(T,B)$—for example, the region $\{(\delta,B): \delta>0 \text{ and } B>\alpha, \forall \alpha\in C\}$ where $C=\{0.6,0.8,1.0\}$—where energy extraction is possible despite the absence of quantum correlations in $\bm{\rho}^{AB}$. illustrated in figure \ref{fig:2} and figure \ref{fig:3}.}

\onecolumngrid

\begin{figure}[!h]
    \centering
    \includegraphics[width=\linewidth]{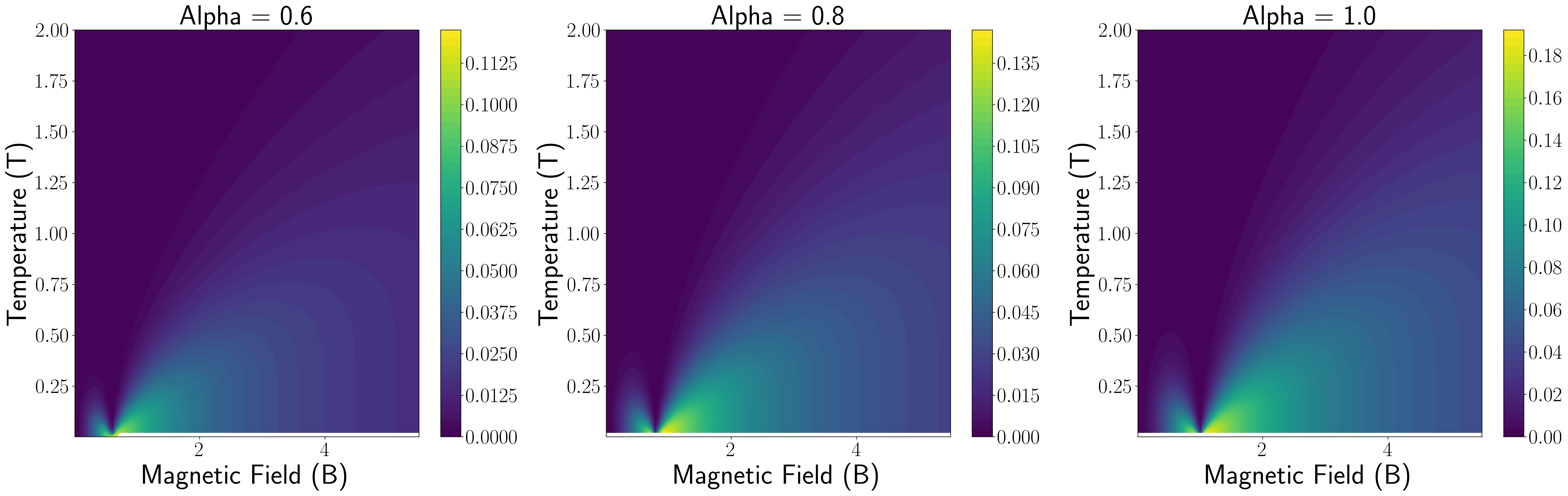}
    \caption{The energy extraction by Bob as a function of temperature ($T$) and magnetic field ($B$) is depicted. Even beyond the critical temperatures $T_c = 0.68$, $0.9$, and $1.13$ for $\alpha = 0.6$, $0.8$, and $1.0$, respectively—where neither entanglement nor \textcolor{black}{discord} exists in $\bm{\rho}^{AB}$—Bob can still extract some energy. \textcolor{black}{Moreover, at the critical value $\alpha = B$, where a phase transition occurs in the ground state from an entangled state to a product state, entanglement remains nonzero, yet no energy can be extracted.}
}
    \label{fig:1}
\end{figure}


\newpage

\twocolumngrid
\section{Quantum correlation analysis\label{sec:level3}}
\textcolor{black}{First of all, we might ask ourselves why we are interested in quantum correlations analysis? Hotta studied QET throughout, utilizing quantum correlations, particularly entanglement and discord between two regions. In this sense both discord and entanglement plays an important roles for QET. However to show that both discord and entanglement are not necessary to achieve QET, in this section, we analyze these correlations.} In general, it is very difficult to quantify the measure of entanglement for any arbitrary given state. For a pure bipartite quantum state, entanglement entropy is one of the ``good'' measures of entanglement. In other cases, many measures, such as the Convex-roof extension\cite{Vidal_2000}, Squashed Entanglement\cite{Alicki_2004, Brand_o_2011, Christandl2004}, Entanglement of Formation\cite{Bennett1996}, Entanglement Cost\cite{Bennett1996Cost}, and Relative Entropy of Entanglement\cite{Nielsen2001}, are used to capture different aspects of this quantum correlation.
Moreover, entanglement measures such as the entanglement of formation\cite{PhysRevLett.80.2245,Vidal_2000,Vidal_2002} require optimization procedures over all possible pure state decomposition of the mixed state, which is computationally expensive. These measures, while providing valuable insights, are not always practical for large systems or for real-time applications.

So, how do we quantify entanglement in mixed state for two qubit system? Fortunately, thanks to Karol Życzkowski, Paweł Horodecki, Anna Sanpera, and Maciej Lewenstein, we have a necessary and sufficient criterion for quantifying entanglement in the case of a two-qubit system, which is called the PPT criterion \cite{Horodecki_1996,Vidal_2002}.Additionally, there is another class of entanglement measure, i.e., Entanglement of Concurrence\cite{PhysRevLett.78.5022}, which is based on the property of entanglement monotone to measure entanglement.

From the equation \eqref{eq:5}, we write the initial state $\bm{\rho}^{AB}$ in the matrix representation with respect to the standard basis as follows:
\begin{align}
    \bm{\rho}^{AB}=\frac{1}{Z}\begin{bmatrix}
    e^{-\beta B} & 0 & 0& 0\\
    0&\cosh(\beta \alpha)&-\sinh(\beta \alpha)&0\\
    0&-\sinh(\beta \alpha)& \cosh(\beta\alpha)&0\\
    0&0&0&e^{\beta B}
\end{bmatrix}\label{eq:26}
\end{align}

Now in the following subsections, we will explore the entanglement analysis in terms of entanglement of concurrence and PPT criterion.

\subsection{Positive partial transpose criterion\label{sec:subIIIA}}
From the equation \eqref{eq:26}, the partial transpose of the initial state $\bm{\rho}^{AB}$ with respect to $B$ is given by
\begin{align}
    \bm{\rho}^{T_B}&=\frac{1}{Z}\begin{bmatrix}
    e^{-\beta B} & 0 & 0& -\sinh(\beta \alpha)\\
    0&\cosh(\beta \alpha)&0&0\\
    0&0& \cosh(\beta\alpha)&0\\
    -\sinh(\beta \alpha)&0&0&e^{\beta B}
\end{bmatrix}\label{eq:27}
\end{align}

Now the PPT criterion states that $\bm{\rho}^{AB}$ is a separable state if and only if all the eigenvalues of $\bm{\rho}^{T_B}$ are non-negative. Hence the contrapositive of PPT criterion is, $\bm{\rho}^{AB}$ is not a separable state if and only if $\bm{\rho}^{T_B}$ is negative i.e., if at least one of the eigenvalues of $\bm{\rho}^{T_B}$ is negative then $\bm{\rho}^{AB}$ entangled. Thus for $\bm{\rho}^{T_B}$, we have to find at least one negative eigenvalue to guarantee that $\bm{\rho}^{AB}$ is entangled.
The eigenvalues of $\bm{\rho}^{T_B}$ are given in the following expressions:
\begin{align}
    \lambda_{1,2}&=\frac{\cosh(\beta \alpha)}{Z}\notag\\
    \lambda_{\pm}&=\frac{1 + e^{2\beta B} \pm \sqrt{1 - 4e^{2\beta B} + e^{4\beta B} + 2e^{2\beta B}\cosh(2\beta\alpha)}}{2Ze^{\beta B}}\notag
\end{align}
Here we clearly see that $\lambda_1,\lambda_2,\lambda_+$ are positive eigenvalues. Only $\lambda_-$ can be negative if and only if 
$$1 + e^{2\beta B} - \sqrt{1 - 4e^{2\beta B} + e^{4\beta B} + 2e^{2\beta B}\cosh(2\beta\alpha)}<0$$
Equivalently $\lambda_-$ is negative if and only if $\beta\alpha>\frac{\cosh^{-1}(3)}{2}$, \textcolor{black}{where $\cosh^{-1}$ denotes inverse hyperbolic cosine function}. Consequently we find a critical temperature $T_c$ given by
\begin{align}
    \boxed{T_c:=\frac{2\alpha}{\cosh^{-1}(3)}=\frac{\alpha}{\ln(\sqrt{2}+1)}}\label{eq:28}
\end{align}
Therefore the state $\bm{\rho}^{AB}$ is guaranteed to be entangled for all temperature below the critical temperature $T_c$. Moreover we calculate the negativity\cite{Vidal_2002},which is defined as
$$\mathcal{N}(\bm{\rho}^{AB})=\bigg|\sum_{\lambda_i<0}\lambda_i\bigg|=\sum_i\frac{|\lambda_i|-\lambda_i}{2}$$
Where $\lambda_i$ for all $i$ are the eigenvalues of $\bm{\rho}^{T_B}$. Thus state $\bm{\rho}^{AB}$ is entangled if we have non-zero negativity.

\subsection{Entanglement of concurrence\label{sec:subIIIB}}
The entanglement of concurrence for two qubits mixed state was first proposed by Scott Hill and William K. Wootters \cite{PhysRevLett.80.2245, PhysRevLett.78.5022}, which is based on entanglement of formation\cite{Bennett1996Cost}. The entanglement of concurrence or simply concurrence of a mixed state for a two qubits system is,
$$C(\bm{\rho}^{AB}) \equiv \max(0, \lambda_1 - \lambda_2 - \lambda_3 - \lambda_4)$$
Here $\lambda_1, \ldots, \lambda_4$ are the square root of eigenvalues in decreasing order of the matrix
$\bm{\rho}^{AB}\tilde{\bm{\rho}}^{AB}$. 
where
$$\tilde{\bm{\rho}}^{AB}= (\sigma_y \otimes \sigma_y) \bm{\rho}^{*AB} (\sigma_y \otimes \sigma_y)$$
the spin-flipped state of $\bm{\rho}^{AB}$ and $\sigma_y$ a Pauli spin matrix. The complex conjugation $^*$ is taken in the eigenbasis of the Pauli matrix $\sigma_z$. The square root of eigenvalues of $\bm{\rho}^{AB}\tilde{\bm{\rho}}^{AB}$ are given by
\begin{align}
    \lambda_1&=\frac{1+\sinh(\beta\alpha)}{Z}\notag\\
    \lambda_2&=\lambda_3=1/Z\notag\\
    \lambda_4&=\frac{1-\sinh(\beta\alpha)}{Z}\notag
\end{align}
Thus $\lambda_1 - \lambda_2 - \lambda_3 - \lambda_4)=\frac{2}{Z}(\sinh(\beta\alpha)-1)$, Therefore the concurrence for the state $\bm{\rho}^{AB}$ is
\begin{align}
    \boxed{C(\bm{\rho}^{AB})=\max\bigg\{0,\frac{2}{Z}(\sinh(\beta\alpha)-1)\bigg\}}\label{eq:29}
\end{align}
Now, we observe that the concurrence in equation \eqref{eq:29} vanishes if $\sinh(\beta\alpha)-1)<0$ if and only if $\beta\alpha<\sinh^{-1}(1)=\ln(1+\sqrt{2})$.
Therefore we get the critical temperature $T_c$, defined as

\begin{align}
    \boxed{T_c=\frac{\alpha}{\sinh^{-1}(1)}=\frac{\alpha}{\ln(1+\sqrt{2})}}\label{eq:30}
\end{align}
Thus, above the critical temperature $T_c$, the state $\bm{\rho}^{AB}$ has no concurrence. Equivalently, $\bm{\rho}^{AB}$ becomes a separable state. From equations \eqref{eq:28} and \eqref{eq:30}, we observe that both the entanglement of concurrence and the PPT criterion for the entanglement analysis give us the same critical temperature $ T_c $ as expected \cite{Frank_Verstraete_2001}. It is interesting to see that the critical temperature is independent of the magnetic field $B$.
\subsection{Quantum Discord in a Two-Qubit Thermal State\label{sec:III.C}}

In this section we analyze the quantum discord of a two-qubit thermal state\eqref{eq:26}. Let us denote the elements in our initial density matrix as
\begin{equation*}
a = \frac{e^{-\beta B}}{Z},
d = \frac{e^{\beta B}}{Z},
w = \frac{\cosh (\beta \alpha)}{Z},
z = -\frac{\sinh (\beta \alpha)}{Z}
\end{equation*}

The quantum discord is defined as the difference between the total correlations (given by the mutual information) and the classical correlations obtained from an optimal measurement on subsystem $A$\cite{PhysRevLett.88.017901}. For the present X state, the discord (when the measurement is performed on $A$) can be written in closed form as \cite{PhysRevA.81.042105, PhysRevA.77.042303}:
\begin{equation}
D(\rho)= S(\tilde\rho_A) - S(\rho) + \min\{S_1, S_2\}
\label{eq:d.31}
\end{equation}
where the von Neumann entropy is defined as
\begin{equation*}
S(\rho) = -\sum_i \lambda_i \log_2 \lambda_i
\end{equation*}
where $\lambda_i$ are eigenvalues of $\rho$ and the marginal state of $A$ is obtained by tracing out subsystem $B$, thus in our case:
\begin{equation}
\tilde\rho_A = \mathrm{Tr}_B \, \bm{\rho}^{AB} = \begin{pmatrix}
a+w & 0 \\[1mm]
0 & w+d
\end{pmatrix}\label{eq:d.32}
\end{equation}
Therefore,
\begin{equation}
S(\tilde\rho_A) = - (a+w) \log_2 (a+w) - (w+d) \log_2 (w+d)\label{eq:d.33}
\end{equation}

The total state $\rho=\bm{\rho}^{AB}$ has eigenvalues $a$, $d$, $w+z$, and $w-z$; hence,
\begin{align}
    S(\rho) = -[a\log_2 a +& d\log_2 d + (w+z) \log_2 (w+z)\nonumber\\
    &+ (w-z) \log_2 (w-z)]\label{d.34}
\end{align}

The two candidate expressions for the conditional entropy after a local measurement on $A$ are given by
\begin{align}
S_1 &= - a \log_2\frac{a}{a+w} - w \log_2\frac{w}{a+w} \nonumber\\
&\hspace{2cm}- d \log_2\frac{d}{d+w} - w \log_2\frac{w}{d+w}, \label{eq:d.35} \\
S_2 &= -\frac{1}{2}(1+\Gamma) \log_2 \frac{1+\Gamma}{2} -\frac{1}{2}(1-\Gamma) \log_2 \frac{1-\Gamma}{2}, \label{eq:d.36}
\end{align}
with
\begin{equation*}
\Gamma = \sqrt{(a-d)^2 + 4z^2}
\end{equation*}
Equation\eqref{eq:d.31} then provides the full analytic expression for the quantum discord of the thermal state before any measurement is performed on $\bm{\rho}^{AB}$.

To illustrate the dependence of quantum discord on temperature $T$ and magnetic field $B$ for fixed coupling constants $\alpha$, we have computed Eq.\eqref{eq:d.31} numerically. Figure \ref{fig:3} displays contour plots of the analytic discord for three representative values $\alpha = 0.6,\, 0.8,\, 1.0$.

As seen in Fig. \ref{fig:3}, the discord exhibits nontrivial dependence on both $T$ and $B$ and eventually it vanishes after certain temperature $T$ and magnetic field $B$.

\subsection{Discord After a Local Measurement}

We note that if a nonselective measurement in the $\sigma_x$ basis is performed on subsystem $A$, thus from equation \eqref{eq:6}, the initial state is transformed in the Alice's $\ket{\pm}$ basis to
\begin{equation}
\bm{\rho}^{AB}=\sum_{i,j\in \{ \pm \}} \ket{i}\bra{j}\ot \rho^B_{ij}\label{eq:d.37}
\end{equation}
After Alice's measurements with $M(k)=\frac{I\pm \sigma_x}{2}\otimes I=\ket{k}\bra{k}\ot I$ with $k=\pm$, state becomes
\begin{equation}
\rho_A=\sum_k \ket{k}\bra{k}\ot \rho^{B}_{kk}\label{eq:d.38}
\end{equation}
Such a measurement renders the state classical with respect to $A$ (i.e., diagonal in the $\sigma_x$ basis), thus it is in the classical-quantum (CQ) state form and therefore the quantum discord for the post-measurement state vanishes,
\begin{equation}
D(\rho_A)=0.
\end{equation}

\onecolumngrid

\begin{figure*}[!ht]
    \centering
    \includegraphics[width=0.9\textwidth]{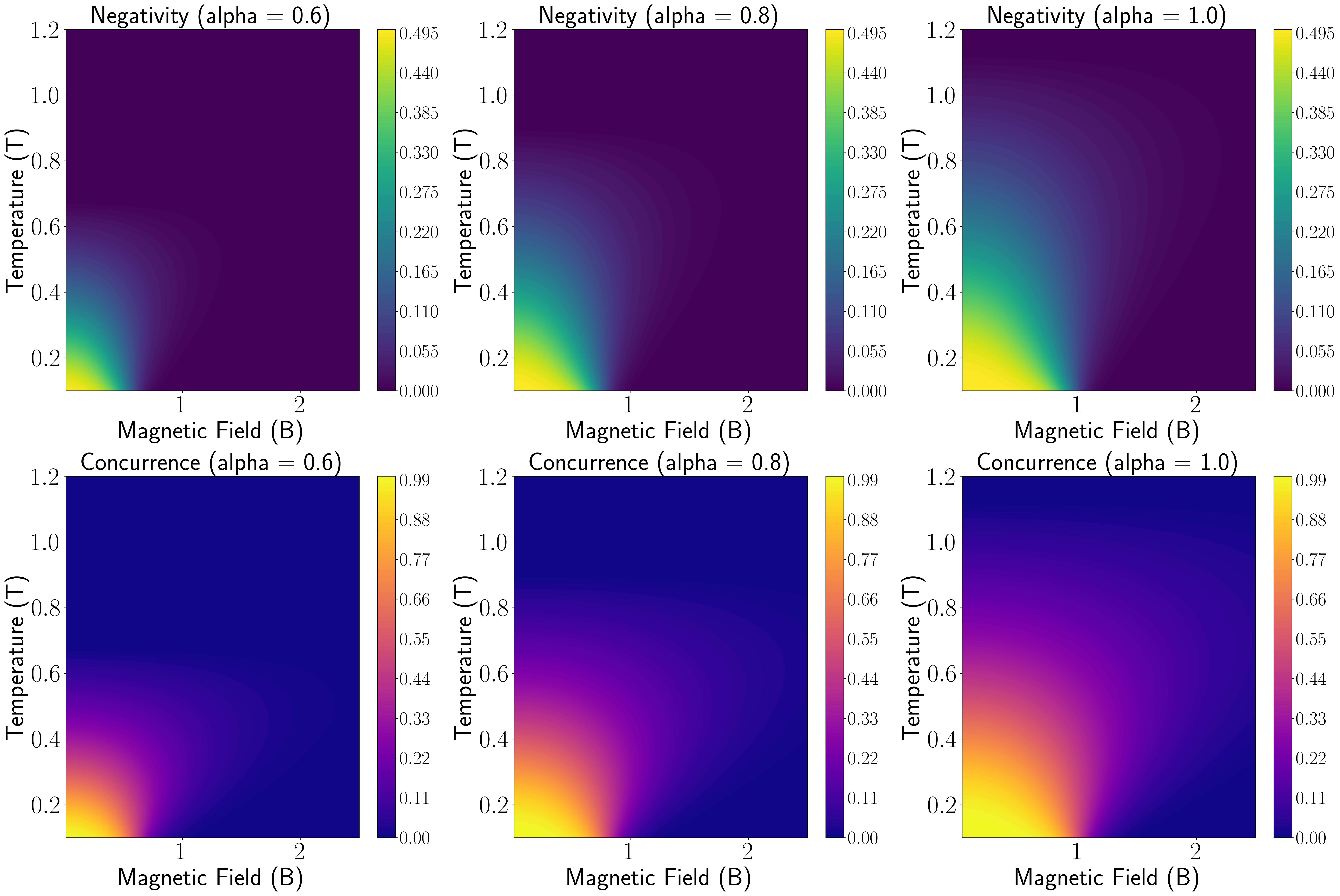}
    \caption{The contour plots of negativity and concurrence for various values of $\alpha$ are illustrated. For each fixed value of either $T$ or $B$, the color gradient decreases, indicating a corresponding decrease in entanglement as either $B$ or $T$ increases. Additionally, both concurrence and negativity exhibit similar characteristics. At low temperatures and magnetic fields, the state $\bm{\rho}^{AB}$ possesses maximum entanglement. As $T \to 0$ and $B < \alpha$, the ground state becomes a Bell pair, $\ket{\Psi^-}$, which is a maximally entangled state. \textcolor{black}{At critical point of ground state degeneracy, $\alpha=B$, there is a nonzero entanglement in $\bm{\rho}^{AB}$}.}
    \label{fig:2}
\end{figure*}

\begin{figure}[!ht]
    \centering
    \includegraphics[width=0.9\linewidth]{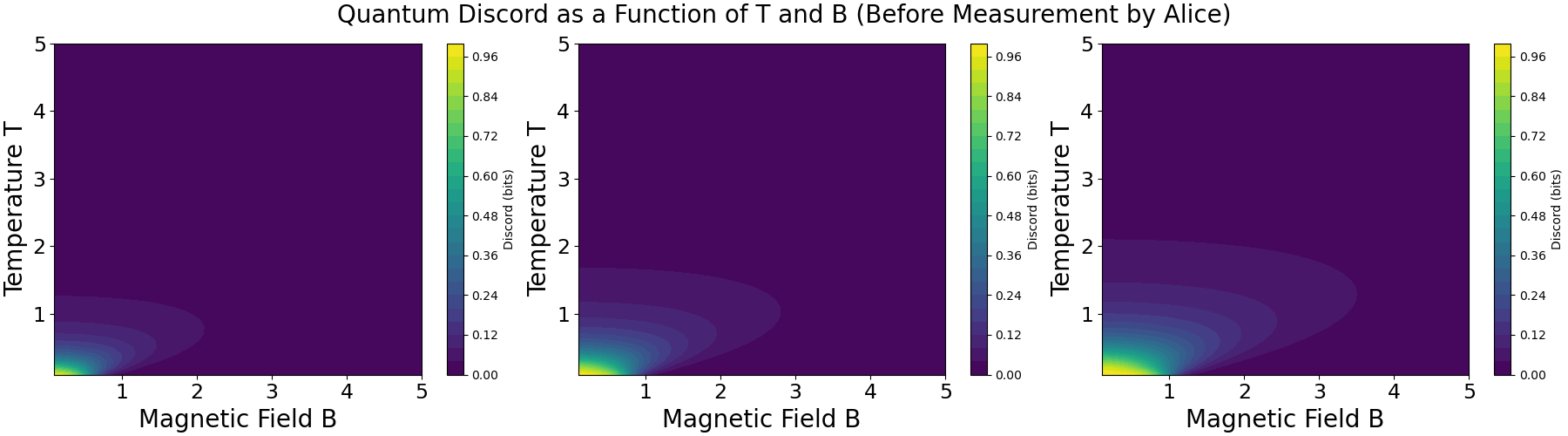}
    \caption{Quantum discord (in bits) as a function of temperature $T$ and magnetic field $B$ for three fixed values of the coupling constant: (a) $\alpha=0.6$, (b) $\alpha=0.8$, and (c) $\alpha=1.0$. The discord is computed using the expression $D(\rho)= S(\rho_A)-S(\rho)+\min\{S_1,S_2\}$ with the parameters $a$, $d$, $w$, and $z$ defined in the subsection \ref{sec:III.C}. }
    \label{fig:3}
\end{figure}

\newpage
\twocolumngrid

\section{QET with excited states\label{sec:level4}}
Suppose we perform the Q.E.T. protocol with the entangled state $\ket{e_+} = \frac{\ket{01} + \ket{10}}{\sqrt{2}}$, which is the highest excited state in our case \textcolor{black}{for $B<\alpha$. On the other hand if $B>\alpha$, the state $\ket{e_+}$ becomes second highest excited state. In both cases,} we observe that, in the measurement process conducted by Alice, she can extract an amount of energy $\alpha/2$ instead of locally injecting it into her system. After the LOCC protocol\eqref{eq:19}, Bob can extract a maximum of $E^+_B := \frac{\alpha + \sqrt{\alpha^2 + B^2}}{2}$ amount of energy. In contrast, in the ground state Q.E.T. protocol, Alice has to inject $\alpha/2$ amount of energy to initiate the protocol. Followed by LOCC, Bob can extract a maximum of $E^{G}_B:=\frac{\sqrt{\alpha^2 + B^2} - \alpha}{2}$ amount of energy. Therefore, we observe that the excited state Q.E.T. protocol dominates the ground state Q.E.T. protocol, which dominates the thermal state Q.E.T. protocol.

\section{Quantum Energy Extraction (QEE) with a product state\label{sec:level5}}
Let us rewrite the Hamiltonian from the section \ref{sec:level2} as $H=H_A+H_B+V$ with a choice of the constant term $\epsilon=B$ to make it positive semi-definite in which $H_A=\frac{B}{2}(\mathbf{I}+\sigma^A_z), H_B=\frac{B}{2}(\sigma^B_z+\mathbf{I}) \text{ and }V=\alpha(\sigma^{A}_{+}\ot\sigma^{B}_{-}+\sigma^{A}_{-}\ot\sigma^{B}_{+})$.
We assume that $B>\alpha$. The energy eigenvalues are $\epsilon_{g}=0, \epsilon_{1}=2B, \epsilon_2=B+\alpha,$ and $ \epsilon_3=B-\alpha$. The corresponding eigenstates are $\ket{e_g}=\ket{00}, \ket{e_{1}}=\ket{11}, \ket{e_2}:=\frac{\ket{01}+\ket{10}}{\sqrt{2}}$ and $\ket{e_3}:=\frac{\ket{01}-\ket{10}}{\sqrt{2}}$.\\

Clearly, we see that the ground state is a product state. We show that the energy can be extracted from the ground state by using measurements and LOCC protocol below.
Before we begin the protocol, we state the following lemma, which we use for the rest of the calculation.

\begin{lem}\label{lem:IV.1}
     Suppose $O^A_i$ and $O^B_j$ are local operators i.e., $O^A_i=O_i\ot\mathbf{I}$ and $O^{B}_j=\mathbf{I}\ot O_j$, then these operators commutes i.e., $[O^A_i,O^B_j]=0$
\end{lem}
\begin{proof}
    Clearly, $O^A_i O^B_j=(O_i\ot\mathbf{I})(\mathbf{I}\ot O_j)=(\mathbf{I}\ot O_j)(O_i\ot\mathbf{I})=O^B_jO^A_i$
\end{proof}

\subsection{Measurement on site A}
The state before the measurement is $\bm{\rho}=\ket{e_g}\bra{e_g}$. Suppose, Alice use the same measurement procedure as in the section \ref{sec:subB}. Using equation \eqref{eq:2}, the energy injected by Alice in the measurement process is,

\begin{widetext}
\begin{align}
    E_A&=\sum_k\braket{e_g|M_A(k)HM_A(k)|e_g}=\sum_k\underbrace{\braket{e_g|M_A(k)H_AM_A(k)|e_g}}_{\text{Energy at Alice's site}}+\underbrace{\braket{e_g|M_A(k)H_BM_A(k)|e_g}}_{\text{Energy at Bob's site}}+\underbrace{\braket{e_g|M_A(k)VM_A(k)|e_g}}_{\text{Interaction energy}}\label{eq:31}
\end{align}
\end{widetext}
To make sure that the energy injected by Alice raise the energy of the whole system but does not affect the energy at Bob's site, we explicitly compute the energy at each site after the measurement.
\subsubsection{Energy at Alice's site after measurements}
Notice that $[H_A,M_A(k)]=\frac{iBk}{2}\sigma^A_y$. Moreover, we found $M_A(k)H_AM_A(k)=M_A(k)H_A+\frac{iBk}{2}M_A(k)\sigma^A_y$
Therefore, energy at Alice's site becomes:
\begin{align*}
   E^M_A:&= \sum_k\braket{e_g|M_A(k)H_AM_A(k)|e_g}\\
   &=\sum_k\cancelto{0}{\braket{e_g|M_A(k)H_A|e_g}}+\frac{iBk}{2}\braket{e_g|M_A(k)\sigma^A_y|e_g}\\
    &=\frac{B}{2}\sum_kk\braket{0|M_A(k)|1}=\frac{B}{2}
\end{align*}
Thus, after the measurements, there is an increment of energy density at Alice's site by an amount $B/2$ unit.
\subsubsection{Energy at Bob's site due to Alice's measurements }
Using lemma \ref{lem:IV.1} and the property $M_A^2(k)=M_A(k)$, we see that
\begin{align*}
    E^M_B:=\sum_k\braket{e_g|M_A(k)H_BM_A(k)|e_g}\\
    =\sum_k\braket{e_g|M_A(k)H_B|e_g}=0
\end{align*}
Therefore, energy at Bob's site before and after the measurement of Alice's spin is $0$.
\subsubsection{Interaction energy after Alice's measurements}
Let us calculate $[V,M_A(k)]$, to calculate the interaction energy $E^M_V=\tr(V\rho_A)$ due to the measurement by Alice.

\begin{align*}
    [V,M_A(k)]&=\alpha[\sigma^{A}_{+}\ot\sigma^{B}_{-}+\sigma^{A}_{-}\ot\sigma^{B}_{+},M_A(k)]\\
    &\hspace{1cm}=-i\frac{\alpha k}{2}\sigma^A_z\ot\sigma^B_y
\end{align*}

Which immediately implies $M_A(k)VM_A(k)=M_A(k)V-i\frac{\alpha k}{2}M_A(k)\sigma^A_z\ot\sigma^B_y$. Moreover
\begin{align*}
    \braket{e_g|M_A(k)VM_A(k)|e_g}&=\cancelto{0}{\braket{e_g|M_A(k)V|e_g}}\\
    &-i\frac{\alpha k}{2}\braket{0|M_A(k)\sigma^A_z|0}\cancelto{0}{\braket{0|\sigma^B_y|0}}=0
\end{align*}
Thus, we observe that the interaction energy is $$E^M_V=\sum_k\braket{e_g|M_A(k)VM_A(k)|e_g}=0$$
Therefore, using equation \eqref{eq:31}, we conclude that Alice's measurement increases the energy in the whole system solely due to the energy density at her site, i.e., $E_A = E^M_A + E^M_B + E^M_V = \frac{B}{2}$. Consequently, measurements at site A only affect the energy at site A; equivalently, Bob's site remains unaffected by Alice's measurement.

\subsection{LOCC}
Based on the measurement outcome from Alice, Bob performs the same local operation as we have seen in section \ref{sec:subC}. The state after the local operation is
\begin{align}
\rho_B=\sum_kU_B(k)M_A(k)\ket{e_g}\bra{e_g}M^\dagger_A(k)U^\dagger_B(k)\label{eq:32}
\end{align}

The energy through Bob's action is given by $E_B=\tr(H\rho_B)$. Moreover,
\begin{align}
    E_B=\sum_k\braket{e_g|M^\dagger_A(k)U^\dagger_B(k)HU_B(k)M_A(k)|e_g}\label{eq:33}
\end{align}
Now, we evaluate $\Delta E_{\text{tel}}:=E_B-E_A$ to see whether Bob can extract some energy or not. Clearly,
\begin{align}
    \Delta E_{\text{tel}}=\sum_k\braket{e_g|M_A(k)\big[U^\dagger_B(k)HU_B(k)-H\big]M_A(k)|e_g}\label{eq:34}
\end{align}

As,we have $M^\dagger_A(k)=M_A(k)=\frac{1}{2}(\mathbf{I}+k\sigma^A_x)$ and $U_B(k)=\cos(\theta)\mathbf{I}+ik\sin(\theta)\sigma^B_y$.
To simplify the equation\eqref{eq:34}, we evaluate $U^\dagger_B(k)HU_B(k)-H$. Furthermore, we calculate the commutator $[H,U_B]=ik\sin(\theta)[H,\sigma^B_y]$-- which immediately follows:
\begin{align}
   U^\dagger_B(k)HU_B(k)-H= ik\sin(\theta)U^\dagger_B(k)[H,\sigma^B_y]\label{eq:35}
\end{align}
Notice that, $[H,\sigma^B_y]=[H_A,\sigma^B_y]+[H_B,\sigma^B_y]+[V,\sigma^B_y]$. Using lemma \ref{lem:IV.1}, $[H_A,\sigma^B_y]=0$. Also notice that,
$[H_B,\sigma^B_y]=-iB\sigma^B_x$ and $[V,\sigma^B_y]=i\alpha\sigma^A_x\ot\sigma^B_z$.
Therefore,
\begin{align}
    [H,\sigma^B_y]=-iB\sigma^B_x+i\alpha\sigma^A_x\ot\sigma^B_z\label{eq:36}
\end{align}
Now, plugging equation \eqref{eq:36} in the equation\eqref{eq:35}, we get;
\begin{align}
    U^\dagger_B(k)&HU_B(k)-H\notag\\
    &= ik\sin(\theta)\big(-iBU^\dagger_B(k)\sigma^B_x +i\alpha \sigma^A_x\ot U^\dagger_B(k)\sigma^B_z\big)\label{eq:37}
\end{align}
Again using the lemma \ref{lem:IV.1}, equation\eqref{eq:37}, $M_A^2(k)=M_A(k)$ and $[M_A(k),\sigma^A_x]=0$, we get

\begin{align}
    &M_A(k)\big[U^\dagger_B(k)HU_B(k)-H\big]M_A(k)=\notag\\ &ik\sin(\theta)\big(-iBM_A(k)U^\dagger_B(k)\sigma^B_x +i\alpha M_A(k)\sigma^A_x\ot U^\dagger_B(k)\sigma^B_z\big)\label{eq:38}
\end{align}

Therefore, from the above equation and the equation\eqref{eq:34}, we get $E_B-E_A=\Delta E_{\text{tel}}$ as below:
\begin{align}
   \Delta E_{\text{tel}}&=B\sin(\theta)\sum_kk\braket{0|M_A(k)|0}\braket{0|U^\dagger_B(k)\sigma^B_x|0}\notag\\
   &\hspace{1cm}-\alpha\sin(\theta)\sum_kk\braket{0|M_A(k)\sigma^A_x|0}\braket{0| U^\dagger_B(k)\sigma^B_z|0}\notag
\end{align}
\begin{align}
   &=B\sin(\theta)\sum_kk\braket{0|M_A(k)|0}\braket{0|U^\dagger_B(k)|1}\notag\\
   &\hspace{1cm}+\alpha\sin(\theta)\sum_kk\braket{0|M_A(k)|1}\braket{0| U^\dagger_B(k)|0}\notag\\
   &=B\sin(\theta)\sum_k\frac{k^2}{2}\sin(\theta)+\alpha\sin(\theta)\sum_k\frac{k^2}{2}\cos(\theta)\notag\\
   &=\frac{B}{2}(1-\cos(2\theta))+\frac{\alpha}{2}\sin(2\theta)\label{eq:39}
\end{align}
Clearly, the minimum of $E_B-E_A$ from the above equation is given by
$$\boxed{\min_\theta(E_{\text{tel}})=\frac{B-\sqrt{B^2+\alpha^2}}{2}<0}$$
with the choice of $\sin(2\theta)=\frac{-\alpha}{\sqrt{\alpha^2+B^2}}$ and $\cos(2\theta)=\frac{B}{\sqrt{\alpha^2+B^2}}$

The minimum of $(E_{\text{tel}})<0 $ indicates that Bob loses energy from his system. Therefore, Bob can extract maximum $\frac{\sqrt{B^2+\alpha^2}-B}{2}$ unit of energy. This shows that, unlike the \textcolor{black}{usual QET protocol, which requires quantum correlations in the initial state}, one can \textcolor{black}{also extract energy by relaxing the requirement of quantum correlations in the initial state, as in the case of a product ground state here.}

\section{Conclusions}
This investigation into the Quantum Energy Teleportation (QET) protocol across different quantum states and \textcolor{black}{various parametric regimes} reveals significant insights of the role of \textcolor{black}{quantum correlations such as discord and }entanglement as a resource. \textcolor{black}{In the $B<\alpha$ regime,} we found that as ``thermal entanglement'' as well as discord resource decreases, Bob's energy extraction also decreases over temperature. \textcolor{black}{On the other hand, for $B > \alpha$, energy extraction decreases even though entanglement and discord initially increases with temperature before decreasing again. Moreover, there exist regions in $(T,B)$ where energy extraction is possible despite the absence of quantum correlations.}  
In section \ref{sec:level4}, we observed that in the excited state QET protocol with $\ket{e_+}=\ket{\psi^{+}}$, not only allows Bob to extract higher energy but also enables Alice to extract energy during the measurement process. This contrasts with the ground-state QET protocol, where Alice must inject energy to initiate the process, and Bob's energy extraction is based on the use of LOCC. Moreover, for the first time, our analysis in Section \ref{sec:level5} shows that, although one cannot directly extract energy from the ground state\textemdash which is a product state\textemdash using any single-step quantum operation such as quantum measurements alone, it is possible to extract energy from the ground state by activating it through measurements followed by LOCC, which constitutes the QET protocol. It is important to note that the QET assumption, which requires the time scale for energy diffusion after local measurements to be much longer than the time scale for LOCC, is maintained. Therefore, although, in principle, we use all the components of the QET protocol\textemdash namely, quantum measurements and LOCC\textemdash we cannot call it QET because the ground state under consideration is a product state. Hence, we refer to it as Quantum Energy Extraction (QEE).
\section{Acknowledgments}
I would like to thank Prof. Kehrein, Prof. Mani and Prof. Usha for their valuable comments and discussions for this work.

\bibliography{apssamp}

\onecolumngrid
\appendix
\section{Proof of optimum energy teleportation\label{sec:A}}
On taking the first derivative of $F(t)$ and setting it to zero, we get the critical point $t=t_0$ for $F(t)$, thus
\begin{align}
    F'(t)&=\big(B\sin(t)+\alpha\cos(t)\big)\frac{\sinh(\beta B)}{Z}+\big(\alpha\sin(t)-B\cos(t)\big)\frac{\sinh(\beta \alpha)}{Z}=0\notag\\
    &\hspace{1.5cm}\implies\tan(t_0)=\frac{B\sinh(\beta \alpha)-\alpha\sinh(\beta B)}{B\sinh(\beta B) + \alpha\sinh(\beta \alpha)}=:\frac{q}{p}\label{eq:A1}
\end{align}
Moreover, we see that the second derivative of $F(t)$ given by
\begin{align}
    F''(t)=\big(B\cos(t)-\alpha\sin(t)\big)\frac{\sinh(\beta B)}{Z}+\big(\alpha\cos(t)+B\sin(t)\big)\frac{\sinh(\beta \alpha)}{Z}\label{eq:A2}
\end{align}
We will show that at $t=t_0$, the second derivative of $F(t)>0$, thus it concludes that indeed the critical point $t=t_0$ is the point of minimum for $F(t)$.
Notice that if $\tan(t_0)=q/p$ which immediately follows that $\cos(t_0)=\frac{p}{\sqrt{p^2+q^2}}$ and $\sin(t_0)=\frac{q}{\sqrt{p^2+q^2}}$. On collecting the like terms in equation \eqref{eq:A2}, we can rewrite $F''(t)$ as
\begin{align}
    F''(t)&=\bigg(B\frac{\sinh(\beta B)}{Z}+\alpha\frac{\sinh(\beta \alpha)}{Z}\bigg)\cos(t)+\bigg(B\frac{\sinh(\beta \alpha)}{Z}-\alpha\frac{\sinh(\beta B)}{Z}\bigg)\sin(t)\notag\\
    &\implies F''(t_0)=p\cos(t_0)+q\sin(t_0)=\frac{1}{\sqrt{p^2+q^2}}[p^2+q^2]=\sqrt{p^2+q^2}>0\label{eq:A3}
\end{align}
In the above equation, we have shown that at $t=t_0$, indeed we achieve the point of minimum for $F(t)=E_B - E_A = p(1 - \cos(t_0)) - q\sin(t_0)$. Now we will show that $F(t_0)<0$ which immediately implies $|F(t_0)|$ is the maximum amount of energy that can be extracted by Bob.
Since, 
\begin{align}
    F(t)&= p(1 - \cos(t_0)) -q\sin(t_0)=p\bigg(1-\frac{p}{\sqrt{p^2+q^2}}\bigg)-\frac{q^2}{\sqrt{p^2+q^2}}\notag\\
    &\hspace{1.7cm}=p-\frac{p^2+q^2}{\sqrt{p^2+q^2}}=p-\sqrt{p^2+q^2}<0\label{eq:A4}
\end{align}
\section{Derivation of optimum L.O.C.C\label{sec:B}}
Any local unitary operation by Bob is just a rotation of spin in the $\text{SU}(2)$ manifold. For any measurement outcome by Alice, i.e., $k=\{\pm 1\}$, the general local unitary operation by Bob can be written as $$U_B(\theta,k,\hat{n})=\cos(\theta)\mathbf{I}+ik(\hat{n}\cdot \Vec{\sigma})\sin(\theta)$$
where $\hat{n}=(n_1,n_2,n_3)$ and $\Vec{\sigma}=(\sigma_x,\sigma_y,\sigma_z)$ are unit vector and Pauli matrices respectively. So we can interpret $U_B(\theta,k,\hat{n})$ as rotation of spin by an angle $\theta$ around the axis $\hat{n}$. Now we will search for an optimum axis of rotation  i.e., $\hat{n}$ such that Bob can extract maximum energy from the QET protocol. We denote all the states from equation \eqref{eq:9}-\eqref{eq:12} after Bob's action by $\ket{\Psi_{B,n}(k,\hat{n})}:=U_B(\theta,\hat{n})\ket{\Psi_{A,n}(k)}$. To proceed further, we have list out all the states after the L.O.C.C. below:

\begin{align}
    \ket{\Psi_{B,00}(k,\hat{n})}&=\frac{(\cos(\theta)-ikn_3\sin(\theta))}{2}\ket{e_{00}}+i\frac{(n_1-in_2)\sin(\theta)}{2}\ket{e_{11}}+\notag\\
    &+\frac{(k\cos(\theta)+ik(n_1-in_2)\sin(\theta)-in_3\sin(\theta))}{2\sqrt{2}}\ket{e_{+}}\notag\\
    &-\frac{(k\cos(\theta)-ik(n_1-in_2)\sin(\theta)-in_3\sin(\theta))}{2\sqrt{2}}\ket{e_{-}}\label{eq:B1}\\
    \ket{\Psi_{B,11}(k,\hat{n})}&=\frac{(\cos(\theta)+ikn_3\sin(\theta))}{2}\ket{e_{11}}+i\frac{(n_1+in_2)\sin(\theta)}{2}\ket{e_{00}}+\notag\\
    &+\frac{(k\cos(\theta)+ik(n_1+in_2)\sin(\theta)+in_3\sin(\theta))}{2\sqrt{2}}\ket{e_{+}}\notag\\
    &+\frac{(k\cos(\theta)-ik(n_1+in_2)\sin(\theta)+in_3\sin(\theta))}{2\sqrt{2}}\ket{e_{-}}\label{eq:B2}\\
    \ket{\Psi_{B,+}(k,\hat{n})}&=\frac{(\cos(\theta)+in_1\sin(\theta)}{2}\ket{e_{+}}+\frac{i(kn_3-in_2)\sin(\theta)}{2}\ket{e_{-}}+\notag\\
    &+\frac{(k\cos(\theta)+ikn_1\sin(\theta)+i(ikn_2-n_3)\sin(\theta))}{2\sqrt{2}}\ket{e_{00}}\notag\\
    &\frac{(k\cos(\theta)+ikn_1\sin(\theta)-i(ikn_2-n_3)\sin(\theta))}{2\sqrt{2}}\ket{e_{11}}\label{eq:B3}\\
    \ket{\Psi_{B,-}(k,\hat{n})}&=\frac{(\cos(\theta)-in_1\sin(\theta)}{2}\ket{e_{-}}+\frac{i(kn_3+in_2)\sin(\theta)}{2}\ket{e_{+}}+\notag\\
    &-\frac{(k\cos(\theta)+kn_2\sin(\theta)-i(kn_1+n_3)\sin(\theta))}{2\sqrt{2}}\ket{e_{00}}\notag\\
    &+\frac{(k\cos(\theta)-kn_2\sin(\theta)+i(n_3-kn_1)\sin(\theta))}{2\sqrt{2}}\ket{e_{11}}\label{eq:B4}
\end{align}
Now using the above four equations, we calculate $\braket{\ket{\Psi_{B,n}(k}|\bm{H}|\ket{\Psi_{B,n}(k}}$ to calculate the energy due to Bob's action i.e., $E_B=\sum_{k\in\{\pm 1\}}\sum_np_n\braket{\ket{\Psi_{B,n}(k}|\bm{H}|\ket{\Psi_{B,n}(k}}$. Thus

\begin{align}
   \braket{\ket{\Psi_{B,00}(k}|\bm{H}|\ket{\Psi_{B,00}(k}} =&-\braket{\ket{\Psi_{B,11}(k}|\bm{H}|\ket{\Psi_{B,11}(k}}=\frac{B}{4}[(n_1^2+n_2^2-n_3^2)\sin^2(\theta)-\cos^2(\theta)]\notag\\
   &\hspace{2cm}+\frac{\alpha}{4}[n_2\sin(2\theta)-2kn_1n_3\sin^2(\theta)]\label{eq:B5}\\
   \braket{\ket{\Psi_{B,+}(k}|\bm{H}|\ket{\Psi_{B,+}(k}}=&-\braket{\ket{\Psi_{B,-}(k}|\bm{H}|\ket{\Psi_{B,-}(k}}=\frac{\alpha}{4}[(n_1^2-n_2^2-n_3^2)\sin^2(\theta)+\cos^2(\theta)]\notag\\
   &\hspace{2cm}+\frac{B}{4}[n_2\sin(2\theta)+2kn_1n_3\sin^2(\theta)]\label{eq:B6}
\end{align}

For notational simplicity and to avoid complexities, we denote $l':=-\frac{B}{4}[(n_1^2+n_2^2-n_3^2)\sin^2(\theta)-\cos^2(\theta)]-\frac{\alpha n_2}{4}\sin(2\theta)$ and $m'=\frac{\alpha}{4}[(n_1^2-n_2^2-n_3^2)\sin^2(\theta)+\cos^2(\theta)]+\frac{Bn_2}{4}\sin(2\theta)$ from equations \eqref{eq:B5}-\eqref{eq:B6}. Therefore, energy through Bob's action is given by 
$E_B(\theta,\hat{n})=-4l'\frac{\sinh(\beta B)}{Z}-4m'\frac{\sinh(\beta \alpha)}{Z}$ in which summation over linear $k$ vanishes. Moreover, if we denote $\Delta E_{tel}(\theta,\hat{n}):=E_B(\theta,\hat{n})-E_A$ then Bob can extract $|\Delta E_{tel}(\theta,\hat{n})|$ amount of energy. Since $\hat{n}$ and $\theta$ are independent variables, optimization of $\Delta E_{tel}(\theta,\hat{n})$ over $\hat{n},\theta$ is equivalent to optimization of $\Delta E_{tel}(\theta,\hat{n})$ over $\hat{n}$ followed by $\theta$. We minimize
\begin{align}
    \Delta E_{tel}(\theta,\hat{n})=(-4l'+B)\frac{\sinh(\beta B)}{Z}+(-4m'+\alpha)\frac{\sinh(\beta \alpha)}{Z}\label{eq:B7}
\end{align}
subject to the constraint $C:=\{(n_1,n_2,n_3)\in \RR^3: n_1^2+n_2^2+n_3^2=1\}$
\subsection{Analyzing the Critical Points}
Suppose we denote $A=\frac{\sinh(\beta B)}{Z}, K=\frac{\sinh(\beta \alpha)}{Z}, c=1-\cos(t), s=\sin(t)$ and $t=2\theta$. The Lagrangian for the objective function along with constraints is given by:
$$\mathcal{L}(\hat{n}, \lambda) = \frac{BAc}{2} (1+n_1^2+n_2^2-n_3^2) + \frac{\alpha Kc}{2} (1-n_1^2+n_2^2+n_3^2) + \alpha n_2 As - Bn_2 Ks + \lambda (n_1^2 + n_2^2 + n_3^2 - 1)$$
On taking the partial derivatives of $\mathcal{L}$ with respect to $n_1$, $n_2$, $n_3$, and $\lambda$, and set them to zero, we get

\begin{align}
\frac{\partial \mathcal{L}}{\partial n_1} =& BAc n_1 - \alpha Kc n_1 + 2 \lambda n_1 = 0\label{eq:B8}\\
\frac{\partial \mathcal{L}}{\partial n_2} =& BAc n_2 + \alpha Kc n_2 + \alpha As - B Ks + 2 \lambda n_2 = 0\label{eq:B9}\\
\frac{\partial \mathcal{L}}{\partial n_3} =& -BAc n_3 + \alpha Kc n_3 + 2 \lambda n_3 = 0\label{eq:B10}\\
\frac{\partial \mathcal{L}}{\partial \lambda} =& n_1^2 + n_2^2 + n_3^2 - 1 = 0\label{eq:B11}
\end{align}
From the equation \eqref{eq:B8} and \eqref{eq:B10}, we found
$n_1 (BAc - \alpha Kc + 2 \lambda) = 0$
 and $n_3 (-BAc + \alpha Kc + 2 \lambda) = 0$ respectively. We have the following two cases:\\
 Case 1: $n_1 = 0$, $n_3 = 0$ therefore $n_2=\pm 1$\\
 Case 2: $n_1 \neq 0$ or $n_3 \neq 0$\\
In this case, we need to solve the equations $BAc - \alpha Kc + 2 \lambda = 0$
and $-BAc + \alpha Kc + 2 \lambda = 0$. Thus on adding these equations, we get:
$\lambda=0$. Thus, on substitution of $\lambda = 0$ back into the equations \eqref{eq:B8} and \eqref{eq:B10}, we get
$\boxed{BA = \alpha K}$. On the contrary we already have $B<\alpha\implies A<K\implies BA<\alpha K$. Therefore, there is no non-zero solution for $n_1,n_3$.
Now we have two critical points i.e., $\hat{n}\in\{(0,1,0),(0,-1,0)\}$. To find the critical point of absolute minimum for equation \eqref{eq:B7}, we need to evaluate $\Delta E_{tel}(\theta,\hat{n})$ for these critical points and find the minimum.
Let's plug in $\hat{n}=(0,1,0), (0,-1,0)$ and evaluate. We found:
\begin{align}
    \Delta E_{tel}(\theta,(0,1,0))&= (BA+  \alpha K)c - (BK-\alpha  A )s\label{eq:B12}\\
    \Delta E_{tel}(\theta,(0,-1,0))&= (BA+  \alpha K)c + (BK-\alpha  A )s\label{eq:B13}
\end{align}
As hyperbolic function $K>A$, clearly we observe that  $\Delta E_{tel}(\theta,(0,1,0))<\Delta E_{tel}(\theta,(0,-1,0))$ which implies $\hat{n}=(0,1,0)$ is the point of minimum. Thus $U_{B}(\theta,k):=U_B(\theta,k,(0,1,0))$ and $m:=m', l:=l'$. Notice that if $B>\alpha$ then equation \eqref{eq:B13} would be the minimum.

\end{document}